\documentclass[10 pt, letterpaper]{article}

\usepackage{amsmath}
\usepackage{amssymb}
\usepackage{graphicx}
\usepackage{amsmath}
\usepackage{amssymb}
\usepackage{graphicx}
\usepackage{makeidx}
\usepackage{algorithm}
\usepackage{algpseudocode}
\usepackage{epstopdf}

\usepackage{times}
\usepackage[T1]{fontenc}

\newcommand{\sgn}{\mbox {\rm sgn }}

\newcommand{\eal}{\end{align}}
\newcommand{\beq}{\begin{equation}}
\newcommand{\eeq}{\end{equation}}
\newcommand{\bea}{\begin{eqnarray}}
\newcommand{\eea}{\end{eqnarray}}
\newcommand{\beas}{\begin{eqnarray*}}
\newcommand{\eeas}{\end{eqnarray*}}
\newcommand{\ba}{\begin{array}}
\newcommand{\ea}{\end{array}}

\newcommand{\bit}{\begin{itemize}}
\newcommand{\eit}{\end{itemize}}
\newcommand{\ben}{\begin{enumerate}}
\newcommand{\een}{\end{enumerate}}

\newcommand{\bR}{{\mathbb R}}
\newcommand{\bZ}{{\mathbb Z}}

\newtheorem{proof}{Proof}
\newtheorem{assumption}{Assumption}

\newtheorem{remark}{Remark}[section]
\newtheorem{lemma}{Lemma}[section]

\allowdisplaybreaks[3]

\makeatletter

\usepackage{makeidx}\usepackage{times}

\begin{document}
\date{}

\title{Adaptive model predictive control with exploring property for constrained linear systems that uses basis function model parametrization\thanks{This research has received funding from the European Union Seventh Framework Programme (FP7/2007-2013) under grant agreement n. PIOF-GA-2009-252284 - Marie Curie project ``Innovative Control, Identification and Estimation Methodologies for Sustainable Energy Technologies''.}}
\author{M. Tanaskovic, L. Fagiano, R. S. Smith and  M. Morari \thanks{Automatic Control Laboratory, Swiss Federal Institute of Technology Zurich, Switzerland. Email adresses: tmarko@control.ee.ethz.ch (Marko Tanaskovic),  fagiano@control.ee.ethz.ch (Lorenzo Fagiano), rsmith@control.ee.ethz.ch (Roy Smith), morari@control.ee.ethz.ch (Manfred Morari). Corresponding author M.~Tanaskovic.}}

\maketitle

\section{Introduction}
This manuscript contains technical details of recent results developed by the authors on adaptive model predictive control for constrained linear systems that exhibits exploring property and uses basis function model parametrization.

\section{Problem formulation}

\label{S:problem}

We consider a MIMO, discrete time, strictly proper, linear time invariant (LTI) system with $n_u$ inputs and $n_y$ outputs. The system is known to be stable, but the exact system's dynamics are not known. We denote the vector of control inputs at time step $t\in \bZ$ by $u(t)=[u_1(t),\hdots,u_{n_u}(t)]^T$, where $u_i(t)\in \bR,\,i=1,\hdots,n_u$ are the individual plant inputs at the time step $t$ and $^T$ stands for the matrix transpose operator. In addition, we denote the vector of plant outputs by $y(t)=[y_1(t),\hdots,y_{n_y}(t)]^T$, where $y_j(t)\in \bR ,\, j=1,\hdots, n_y$ are the individual plant outputs, and we denote the vector of output disturbances by $d(t)=[d_1(t),\hdots,d_{n_y}(t)]^T$, where each $d_j(t)\in \bR, \, j=1,\hdots,n_y$ denotes the contribution of the disturbances to the output $j$ at time step $t$.\\
The dynamic relation between the plant inputs and each of the plant outputs $y_j(t),\, j=1,\hdots,n_y$, is given by:
\beq\label{Eq:plant_output}
\begin{aligned}
y_j(t)&=\sum\limits_{i=1}^{n_u}\sum\limits_{k=1}^{m}h_{ji}(k)\zeta\left(\mathcal{L}_k(a,q),u_i(t)\right)+d_j(t)= H_j^T\varphi(t)+d_j(t),
\end{aligned}
\eeq
where $\mathcal{L}_k(a,q),\, k=1,\hdots, m$ are the basis transfer functions defined by the parameter $a$ that are selected by the control designer, $q$ is the time shift operator ($qu(t)\doteq u(t+1)$) and the operator $\zeta\left(\mathcal{L}_k(a,q),u(t)\right)$ denotes the output of a linear system represented by the transfer function $\mathcal{L}_k(a,q)$ at the time step $t$, when the signal $u$ is applied as its input:
\beq\label{Eq:operator}
\zeta\left(\mathcal{L}_k(a,q),u(t)\right)\doteq \sum\limits_{l=1}^{\infty}\Psi_k(a,l)u(t-l),
\eeq
where $\Psi_k(a,l)$ are the impulse response coefficients of the basis function $\mathcal{L}_k(a,q)$. $\varphi(t)\doteq [\zeta\left(\mathcal{L}_1(a,q)u_1(t)\right), \hdots,\\\zeta\left(\mathcal{L}_m(a,q)u_1(t)\right),\hdots, \zeta\left(\mathcal{L}_1(a,q)u_{n_u}(t)\right), \hdots, \zeta\left(\mathcal{L}_m(a,q)u_{n_u}(t)\right)]^T$ is the regressor vector and each of the vectors $H_j\in \bR^{n_um},\, j=1,\hdots, n_y$ contains $n_ym$ modeling coefficients needed to describe the influence of all the control inputs on the plant output $j$: $H_j\doteq[h_{j1}(1), \hdots, h_{j1}(m), \hdots, h_{jn_u}(1), \hdots, h_{jn_u}(m)]^T$. By defining the matrix $H\in \bR^{n_y\times n_um}$ as $H \doteq \left[H_1, \hdots, H_{n_y}\right]^T$, the dependence of the plant output on the regressor and the disturbance vectors can be written as:
\beq\label{Eq:plant_output_matrix}
y(t)=H\varphi(t)+d(t).
\eeq
\begin{remark}\label{R:different}
Note that the same value of the parameter $a$ and the number of modeling coefficients $m$ is assumed for all the input-output pairs, in order to simplify the notation. All the results can easily be extended to the case when different values for $a$ and $m$ are assumed for each input-output pair.  
\end{remark}
The measured output available for feedback control is corrupted by noise. In particular, the vector of measured plant outputs is given as:
\beq\label{Eq:measured_output}
\tilde{y}(t)=y(t)+v(t),
\eeq
where $v(t)=[v_1(t),\hdots,v_{n_y}(t)]^T$ and $v_j(t),\, j=1,\hdots,n_y$ are the individual measurement noise terms that affect each of the measured plant outputs.
\begin{assumption}\label{A:disturbance}(Prior assumption on disturbance and noise)
$d$ and $v$ are bounded as:
\beq\label{Eq:dist_bound}
\ba{lll}
|d_j(t)|&\leq&\epsilon_{d_j}\\
|v_j(t)|&\leq&\epsilon_{v_j}
\ea,\,\forall t\in\bZ,\,\forall j=1,\hdots,n_y,
\eeq
where $\epsilon_{d_j}$ and $\epsilon_{v_j}$ are positive scalars.\
\end{assumption}
We also use the noise and disturbance magnitude bounds in the vector notation as $\epsilon_d=[\epsilon_{d_1}, \hdots, \epsilon_{d_{n_y}}]^T$ and $\epsilon_v=[\epsilon_{v_1}, \hdots, \epsilon_{v_{n_y}}]^T$.
\begin{assumption}\label{A:system_class}(Prior assumption on the system)
The plant belongs to the following model set: $H\in \mathcal{F}(0)$, with
\beq\label{Eq:system_set}
\mathcal{F}(0) \doteq \left\{H\in \bR^{n_y\times n_um}: A_{j}(0)H_j\leq b_{j}(0), j=1,\hdots, n_y \right\},
\eeq
where the inequalities in \eqref{Eq:system_set} should be interpreted as element-wise inequalities and each matrix $A_{j}(0)\in \bR^{r_{j}(0)\times n_um}$ and vector $b_{j}(0) \in \bR^{r_{j}(0)}$ define a close and convex set, i.e. a polytope with $r_{j}(0)$ faces.
\end{assumption}
According to Assumption \ref{A:system_class}, the initial knowledge about the modeling parameters is that the vectors $H_j, \, j=1,\hdots, n_y$, which form the rows of the matrix $H$, belong to polytopic sets. Note that we initialized the set $\mathcal{F}(0)$ in \eqref{Eq:system_set} at $t=0$ without loss of generality, just to indicate that this is the information available before any measured data is available.\\
Under Assumptions \ref{A:disturbance} and \ref{A:system_class}, the goal is to control the plant in order to track a desired output reference and reject disturbances from $t=0$ up to some finite time step $T$, where the time horizon $T$ can be very large ($T\gg m$). In addition, we assume that the control inputs $u(l), l=-T_1,\hdots,-1$ are known, where the horizon $T_1$ should be long enough such that the effects of $u(-T_1-1)$ on the regressor vector $\varphi(0)$ are small. The error in the regressor vector coming from the initial state can be accounted for by embedding it in the bound on $d(t)$.  Moreover, the controller shall enforce input and output constraints. Such a control objective can be formalized by the following optimization problem:
\begin{eqnarray}\label{Eq:problem_cost}
&
\begin{aligned}
\min\limits_{u(0),\hdots,u(T-1)}\sum_{t=0}^{T} &\left( y(t)-y_{\text{des}}(t) \right)^TQ \left( y(t)-y_{\text{des}}(t) \right)+u(t)^T S u(t) + \Delta u(t)^T R \Delta u(t) \end{aligned}\\\label{Eq:IO_constraints}
&\begin{aligned}
& \text{subject to},\,\forall t\in[0,T] \\
&\ba{lllll}
Cu(t)&\leq &g\\
L \Delta u(t)&\leq&f\\
E y(t)&\leq& p
\ea
\end{aligned}
\end{eqnarray}
where $y_{\text{des}}(t)\in \bR^{n_y}$ is the desired output reference, $Q \in \bR^{n_y\times n_y}$, $S\in \bR^{n_u\times n_u}$ and $R \in \bR^{n_u \times n_u}$ are positive semi-definite weighting matrices selected by the control designer, and $\Delta u(t)=u(t)-u(t-1)$ is the rate of change of the control input. The element-wise inequalities in \eqref{Eq:IO_constraints} define convex sets through the matrices $C\in\bR^{n_i\times n_u}$, $L\in\bR^{n_{\Delta_u}\times n_u}$, $E\in\bR^{n_o\times n_y}$ and the vectors $g\in\bR^{n_i}$, $f\in\bR^{n_{\Delta_u}}$, $p\in\bR^{n_o}$, where $n_i$, $n_{\Delta_u}$ and $n_o$ are the number of linear constraints on the inputs, input rates and outputs, respectively. We assume that the set defining the constraints on $\Delta u(t)$ contains the origin and that the constraint set of $u(t)$ is compact. This assumption is satisfied in most practical problems.
\begin{remark}\label{R:IIR}
Note that the influence of the unmodeled dynamics to the plant outputs can be embedded into the output disturbance vector $d(t)$. The facts that the magnitudes of the control inputs are bounded and that the controlled system is stable can be exploited to calculate the bounds on the contribution of the unmodeled dynamics to the plant outputs. 
\end{remark}
The bounds on the output disturbance and measurement noise, as well as the initial model set $\mathcal{F}(0)$ (see Assumptions  \ref{A:disturbance} and \ref{A:system_class}) are assumed to be known a priory. However, selecting the right model parametrization and estimating the initial bounds on the contribution of the unmodeled dynamics and bounds on the modeling parameters is not a trivial task. In the following section we provide a tutorial on how to select the right basis function parametrization and estimate the initial bounds on modeling parameters and contribution of the unmodeled dynamics.
\section{Selection of the basis functions and estimation of the bounds on the modeling parameter and the contribution of unmodeled dynamics}\label{S:Laguere}
Any stable system can be represented by its impulse response. The impulse response models have the advantage of being simple and straightforward to use. However, depending on the specific application, the required number of impulse response coefficients can be quite large. It is reasonable to expect that in the case when some additional information on the system to be controlled is available, such as the approximate location of the dominant poles, the number of coefficients that are needed to model the system can be significantly reduced. This kind of information can be captured well by using model representations given by orthonormal basis functions.\\
Any stable transfer function can be represented by the following infinite sum:
\beq
G(q)=\sum_{k=1}^{\infty}h(k)\mathcal{L}_k(a,q),
\eeq
where $q$ is the time shift operator ($qu(t)\doteq u(t+1)$) and $\mathcal{L}_k$ are mutually orthogonal transfer functions characterized by:
\beq\label{Eq:ortogonality}
\frac{1}{2\pi j}\oint \mathcal{L}_p(a,q)\mathcal{L}_k(a,q^{-1})\frac{dq}{q}=\delta_{pk},
\eeq
where $\delta_{pk}$ is the Kronecker delta function and the integral is around the unit circle. The parameter $a\in \mathbb{C},\, |a|\leq1$ defines the basis functions. The sequence of coefficients $h(k)$ is convergent as $\lim_{k\rightarrow\infty}h(k)=0$.\\
Basis transfer functions that are usually employed in practice are the Laguerre ones, given as:
\beq\label{Eq:Laguere}
\mathcal{L}_k(a,q)=\frac{\sqrt{1-a^2}}{q-a}\left[\frac{1-aq}{q-a}\right]^{k-1},
\eeq
where $a\in \bR,\, a\in[-1,1]$. If $a$ is selected such that it is close to the dominant pole of the system, only a few coefficients need to be identified in order to have a good model of the system \cite{Laguere}. For systems that exhibit oscillatory behavior, Kautz functions are more appropriate. These functions have the following form:
\beq\label{Eq:Kautz}
\begin{aligned}
\mathcal{L}_{2k-1}(a,q)&=\frac{\sqrt{1\!-\!c^2}(q\!-\!b)}{q^2\!+\!b(c\!-\!1)q\!-\!c}\left[ \frac{-cq^2\!+\!b(c\!-\!1)q\!+\!1}{q^2\!+\!b(c\!-\!1)q\!-\!c}\right]^{k\!-\!1}\\
\mathcal{L}_{2k}(a,q)&=\frac{\sqrt{1\!-\!c^2}(1\!-\!b^2)}{q^2\!+\!b(c\!-\!1)q\!-\!c}\left[ \frac{-cq^2\!+\!b(c\!-\!1)q\!+\!1}{q^2\!+\!b(c\!-\!1)q\!-\!c}\right]_,^{k\!-\!1}
\end{aligned}
\eeq
where $a\in \mathbb{C}, \, |a|\leq1$ and $b=\frac{a+a^*}{1+aa^*}$, $c=-aa^*$ with $^*$ denoting the conjugate of a complex number. Also in this case, the parameter $a$ should be chosen to be close the the dominant oscillatory pole of the system in order to have a system representation that uses only few parameters \cite{Kautz}.\\
The impulse response model can be regarded as a special case of a basis function parametrization. The basis functions that give rise to an impulse response model are:
\beq\label{Eq:impulse_response}
\mathcal{L}_k(a,q)=aq^{-k}
\eeq
In addition, it is possible to use a parametrization that combines the impulse response and Laguerre or Kautz basis functions. In this case the basis functions are defined as: 
\beq\label{Eq:combined}
\mathcal{L}_k(a,q)=\begin{cases} q^{-k} & \text{if}\, k\leq n\\
\mathcal{B}_{k-n}(a,q)q^{-n} & \text{if}\, k>n \end{cases},
\eeq
where $\mathcal{B}_k(a,q)$ are the Laguerre or Kautz basis functions as defined in \eqref{Eq:Laguere} and \eqref{Eq:Kautz} and $n$ is another design parameter that determines the number of the impulse response coefficients to be used in the modeling. Such a parametrization requires a small number of modeling parameters in the case when the system has a transport delay and if $n$ is selected close to the transport delay time  \cite{ydstie}. In addition to the listed basis functions that are the most commonly used in practice, generalized orthonormal basis functions can also be used \cite{orthogonal1}.\\
The main challenge when using a basis functions parametrization is the computation of the initial model set $\mathcal{F}(0)$ and of the bounds on the contribution of the truncated part of the basis function sequence to the plant output. The set $\mathcal{F}(0)$ can be constructed if an upper and a lower bound on each of the coefficients $h_{ji}(k),\,j=1,\hdots,n_y,\,i=1,\hdots,n_u,\, k=1,\hdots ,m$ is known. We present now a possible method to calculate the upper and the lower bound on each of the coefficients and the bound on the contribution of the unmodeled dynamics to the output for a single input-output transfer function, noting that the results can easily be extended to calculating the bounds for the whole MIMO system. We assume that the transfer function from the input $i$ to the output $j$ has the following structure:
\beq\label{Eq:Structure1}
G_{ji}(q)=gq^{-\tau}\frac{\prod\limits_{l=1}^{n_z}(q-z_l)}{\prod\limits_{l=1}^{n_p}(q-p_l)},
\eeq
where $g\in \bR$ is the gain, $\tau \in  \mathbb{N}$, $\tau \geq 0$ is the transport delay, $Z=[z_1,\hdots,z_{n_z}]^T\in \bR^{n_z}$ are the zeros and $P=[p_1,\hdots, p_{n_p}]^T\in \bR^{n_p}$ are the poles of the transfer function.
\begin{remark}
Note that real poles and zeros are selected here in order to simplify the notation. All the derivations can be extended to the case of general complex poles and zeros.
\end{remark}
We assume that the gain, transport delay, poles and zeros are not exactly known, but they belong to the following set:
\beq\label{Eq:set_of_parameters}
\ba{ccccc}
\underline{g}&\leq & g &\leq &\overline{g}\\
\underline{\tau} & \leq & \tau & \leq & \overline{\tau}\\
\underline{Z}&\leq & Z&\leq &\overline{Z}\\
\underline{P} &\leq & P &\leq &\overline{P}
\ea,
\eeq
where the inequalities in \eqref{Eq:set_of_parameters} should be interpreted as element-wise inequalities and the bounds $\underline{g},\, \overline{g}, \, \underline{\tau}, \, \overline{\tau} \in \bR$, $\underline{Z},\, \overline{Z}\in \bR^{n_z}$ and $\underline{P},\, \overline{P}\in \bR^{n_p}$ are assumed to be known. In addition, it is assumed that $\overline{\tau}\geq 0$ and that the bounds on the poles are selected such that all the poles of the system lay inside the unit circle. Under these assumptions, the goal is, for a given values of parameter $a$ and the model order $m$, to find the upper and lower coefficient bounds $\overline{h}_{ji}(k)$ and $\underline{h}_{ji}(k), \, k=1,\hdots, m$ such that it holds:
\beq\label{Eq:coeficient_bounds}
\underline{h}_{ji}(k) \leq h_{ji}(k) \leq \overline{h}_{ji}(k), \, k=1,\hdots, m.
\eeq
To this end, we denote by $\Psi(g,\tau,Z,P,l),\, l=1,\hdots, \infty$ the impulse response coefficients of the transfer function \eqref{Eq:Structure1} for different values of $g$, $\tau$, $Z$ and $P$. Then for any fixed $g$, $\tau$, $Z$ and $P$, the impulse response coefficients $\Psi (g,\tau, Z,P,l)$ can be represented as:
\beq\label{Eq:representation_impulse}
\Psi(g,\!\tau,\!Z,\!P,\!l)=\sum\limits_{k=1}^{\infty}h(g,\!\tau,\!Z,\!P,\!k)\Psi_k(a,\,l), \!\forall l=1,\,\hdots,\, \infty,
\eeq
where, because of the orthogonality of the basis functions, the coefficients $h(g,\tau,Z,P,k)$ are given as normalized projections of $\Psi(g,\tau,Z,P,l)$ on the impulse responses of the basis functions:
\beq\label{Eq:Projection}
h(g,\tau,Z,P,k)=\frac{\sum\limits_{l=1}^{\infty}\Psi(g,\tau,Z,P,l)\Psi_k(a,l)}{\sum\limits_{l=1}^{\infty}\Psi_k(a,l)^2}.
\eeq
Therefore, the upper and the lower bounds \eqref{Eq:coeficient_bounds} can be calculated as:
\beq\label{Eq:finding_bounds}
\begin{aligned}
\overline{h}_{ji}(k)&=\sup\limits_{g,\tau,Z,P} h(g,\tau,Z,P,k)\\
\underline{h}_{ji}(k)&=\inf \limits_{g,\tau,Z,P}h(g,\tau,Z,P,k)\\
&\text{subject to \eqref{Eq:set_of_parameters}},
\end{aligned}
\eeq
The contribution of the unmodeled dynamics on the transfer function from input $i$ to the output $j$ is then given by:
\beq
\eta_{ji}(t)=\sum\limits_{k=m+1}^{\infty}h_{ji}(k)\zeta \left(\mathcal{L}_k(a,q)u_i(t)\right).
\eeq
The contribution of the unmodeled dynamics is guaranteed to be bounded: $|\eta(t)|\leq \overline{\eta_{ji}}, \, \forall t$, where the bound $\overline{\eta_{ji}}$ is given by:
\beq\label{Eq:unmodeled_dynamics_bound}
\begin{aligned}
&\overline{\eta}_{ji}=\overline{u_i} \sup\limits_{g,\tau, Z,P}\sum\limits_{l=1}^{\infty}\left|\Psi(g,\!\tau,\! Z,\! P,\! l)-\sum\limits_{k=1}^mh(g,\!\tau,\! Z,\! P,\! k)\Psi_k(a,\! l)\right|\\
&\text{subject to \eqref{Eq:set_of_parameters}}
\end{aligned}
\eeq
where $\overline{u}_i$ can be calculated as:
\beq\label{Eq:max_u_i}
\overline{u}_i=\max\limits_{Cu\leq g}|u_i|
\eeq
Note that \eqref{Eq:max_u_i} can be computed by solving an LP. However, \eqref{Eq:finding_bounds} and \eqref{Eq:unmodeled_dynamics_bound} are in general very difficult to solve as they are infinite dimensional nonlinear optimization problems. However, the fact that the impulse response coefficients exponentially converge to zero can be used to approximate the infinite sums in \eqref{Eq:Projection} and \eqref{Eq:unmodeled_dynamics_bound} by finite sums of appropriate length. The optimization problems  \eqref{Eq:Projection} and \eqref{Eq:unmodeled_dynamics_bound} would then be transformed into finite dimensional optimization problems that could be tackled by griding the box defined in \eqref{Eq:set_of_parameters}. In this approach the values of $h(g,\tau,Z,P,k),\, k=1,\hdots,m$ would be calculated for each of the grid points and the optimization problems could be approximately solved by selecting the grid point that results in the best value of the cost function as the optimum.\\
Such a solution would be suboptimal and the calculated bounds could be further inflated in order to account for the considered approximations.
\section{Adaptive control algorithm}\label{S:Overall_algorithm}
Since the true plant is not exactly known and its outputs are subject to unknown output disturbances, the optimal control problem \eqref{Eq:problem_cost} can not be exactly solved a priori and a suboptimal approach has to be sought. Therefore, in order to approximately optimize the given control objective, while guaranteeing satisfaction of the constraints \eqref{Eq:IO_constraints}, we propose the use of a receding horizon approach, combined with an adaptive control scheme that aims to improve the knowledge on the system's dynamics over time. In this setting, at each time step a sequence of future control inputs is calculated and only the first element of this sequence is applied to the plant. In particular, to guarantee output constraint satisfaction, we aim to identify, at each time step, the set of all the models that are consistent with the initial assumptions on the real plant and the input-output measurements collected up to that time step (model set). If the prior assumptions are valid, this set is guaranteed to contain also the true plant's dynamics. Then, the control computation is carried out in such a way to ensure that the constraints are satisfied for all the models inside this set, hence also for the actual plant.\\
In order to accomplish the model set identification and the robust control computation, we rely on a recursive SM identification algorithm, and an MPC controller. The identification algorithm is such that the model set can be recursively refined with each new output measurement. In addition to the model set, the identification algorithm also provides a nominal model of the plant at each time step.\\
The control input is calculated by solving an optimal control problem that minimizes the weighted $l_2$ norm of the tracking error for the nominal model over a finite horizon, while at the same time satisfying robustly the constraints \eqref{Eq:IO_constraints}.
\begin{remark}\label{R:why_nominal_model}
Note that an alternative to formulating the cost function based on the tracking error of the nominal model would be to minimize the weighted $l_2$ norm of the worst case tracking error for all the models inside the model set. The drawback of this approach is that if the prior assumptions were not informative enough, the model set would be initially large due to high uncertainty, and then the applied control inputs would be very small in magnitude. Such control inputs would in turn not be very informative, therefore the size of the model set would reduce slowly, hence resulting in a very conservative controller.
\end{remark}
The structure of the described adaptive control algorithm is summarized in Algorithm \ref{Alg:RHCOP2}.
\begin{algorithm}
\begin{itemize}
  \item[1)]At time step $t$, obtain $\tilde{y}(t)$ and update the model set based on the past applied control inputs and measured plant outputs;
\item[2)] Select a nominal model of the plant inside the model set;
  \item[3)]Calculate a sequence of possible future control inputs by solving a finite horizon optimal control problem (FHOCP) that minimizes the weighted $l_2$ norm of the tracking error for the nominal model and enforces input and output constraints for all the models inside the model set;
  \item[4)]Apply the first element of the calculated input sequence, set $t=t+1$, go to 1).
\end{itemize}
\caption{Adaptive MPC algorithm}\label{Alg:RHCOP2}
\end{algorithm}
In the subsections that follow, each of the components of the proposed adaptive control algorithm is described in detail.
\subsection{Real-time set membership identification}\label{S:SM_ID}
We denote the sequence of the input-output data collected up to time step $t$ as:
\beq\label{Eq:measured_data}
\{\varphi(l),\,\tilde{y}(l)\}_{l=0}^{t},
\eeq
where $\varphi(l)\in \bR^{n_um}$ is the regressor vector formed by the control inputs applied up to time step $l-1$, and $\tilde{y}(l)\in \bR^{n_y}$ is the corresponding measured plant output. The regressor vector $\varphi(l)$ can be calculated recursively as:
\begin{equation}\label{Eq:recursive_regressor}
\varphi(l)=W\varphi(l-1)+Zu(l-1),
\end{equation}
where $W\in \bR^{n_um \times n_um}$ and $Z \in \bR^{n_um\times n_u}$ are suitable matrices that depend on the selected basis function parametrization. For completeness, matrices $W$ and $Z$ are given in the Appendix \ref{S:Matrices_reg} for the case of impulse response and Laguerre basis function models.\\
At a given time step $t$, we define the model set $\mathcal{F}(t)$ as the set containing all the matrices $H$ that are consistent with the Assumptions \ref{A:disturbance} and \ref{A:system_class} and the collected input-output data \eqref{Eq:measured_data}:
\beq\label{Eq:FSSm_t}
\begin{aligned}
\mathcal{F}(t)\!\doteq\!
\left\{\ba{ll}
H\!\in\! \mathcal{F}(0)\!:\!&-\epsilon_d\!-\!\epsilon_v\leq \tilde{y}(l)\!-\!H\varphi(l)\leq \epsilon_d\!+\!\epsilon_v,\,\forall l \in [0,t]
\ea\right\}\!\!.
\end{aligned}
\eeq
Each one of the element-wise inequalities in \eqref{Eq:FSSm_t} comes from the fact that the discrepancy between the measured and the predicted values of the output can not exceed the disturbance and noise bounds \eqref{Eq:dist_bound}.\\
Since the initial model set $\mathcal{F}(0)$ is defined by polytopic constraints on each row $H_j^T$ of the matrix $H$, and the constraints in \eqref{Eq:FSSm_t} are linear, the model set $\mathcal{F}(t)$ is still defined by polytopic constraints on $H_j^T,\, j=1,\hdots n_y$. Each of these polytopes can be uniquely described by a set of non-redundant inequalities. Therefore, at a generic time step $t$, the model set $\mathcal{F}(t)$ can be represented as:

\beq\label{Eq:FSS_representation}
\mathcal{F}(t)\!=\! \left\{H\!\in\!\bR^{n_y \times n_um}\!:\!A_j(t)H_{j}\leq b_j(t),\, j=1,\hdots, n_y \right\},
\eeq
where $A_j(t) \in \bR^{r_j(t) \times n_um}$, $b_j(t) \in \bR^{r_j(t)}$ and $r_j(t)$ is the number of non-redundant inequalities pertaining to the $j^\text{th}$ row of the matrix $H$.\\
The matrices $A_j(t)$ and the vectors $b_j(t)$ have to be updated at each time step in order to account for the new measurements. To this end, let us consider the following polytopes:
\beq\label{Eq:definition_of_polytope}
\mathcal{F}_j(t)=\{H_j\in \bR^{n_um}: A_j(t)H_j\leq b_j(t)\},\, j=1\hdots n_y.
\eeq
We note that for each $j$, the polytope $\mathcal{F}_{j}(t)$ can be calculated recursively in time as the intersection of the polytope $\mathcal{F}_{j}(t-1)$ and the two half spaces defined by the newly measured plant output, $\tilde{y}_j(t)$:
\beq\label{Eq:Normal_update}
\begin{aligned}
\mathcal{F}_{j}(t)=
&\mathcal{F}_{j}(t-1) \\
&\cap  \{H_{j}\!\!\in\!\bR^{n_um}\!:\!\varphi(t)^TH_{j}\leq \tilde{y}_j(t)\!+\!\epsilon_{d_j}\!+\!\epsilon_{v_j}\}\\
&\cap  \{H_{j}\!\!\in\!\bR^{n_um}\!:\! -\varphi(t)^TH_{j}\leq -\tilde{y}_j(t)\!+\!\epsilon_{d_j}\!+\!\epsilon_{v_j}\}.
\end{aligned}
\eeq
The matrix $A_j(t)$ and the vector $b_j(t)$ can then be calculated by removing any redundant faces of the polytope $F_j(t)$. This can be done by solving an LP for each face of the polytope, in order to determine whether it is redundant or not \cite{Polytopi1}. A problem of the described recursive update is that the number of faces of $\mathcal{F}_{j}(t)$, $r_j(t)$, can become arbitrarily large,  as in general it grows linearly with time, and hence the memory needed to store $A_j(t)$ and $b_j(t)$ can become impractical. In order to overcome this problem, limited complexity polytopic update algorithm similar to the one proposed in \cite{limited_poly_update} can be used.\\
In addition to the model set, the proposed SM identification algorithm also provides a nominal model of the plant (step 2) of Algorithm \ref{Alg:RHCOP2}). The latter is given by a matrix $H_{c}(t)\in \bR^{n_y\times n_um}$, $H_c=[H_{c,1},\hdots H_{c,n_y}]^T$, where $H_{c,j}(t)\in \bR^{n_um},\, j=1,\hdots, n_y$ are computed as the centers of the maximum volume $l_2$-norm balls inscribed in the polytopes $\mathcal{F}_{j}(t)$. This can be done by solving an LP, however the solution is not unique in general. Therefore, we introduce a regularization term, that penalizes the deviation of the new nominal model from the previous one, giving rise to the following LP:
\beq\label{Eq:central_estimate}
\ba{c}
\max \limits_{\xi_j(t),H_{c,j}(t)} \sum \limits_{j=1}^{n_y}\xi_j(t)-\alpha\|H_{c,j}(t-1)-H_{c,j}(t)\|_1\\
\text{subject to}\\
a_{ji}(t)H_{c,j}(t)\!+\!\xi_j(t)\|a_{ji}(t)\|_2\!\leq\! b_{ji}(t), \, j=1,\hdots, n_y,\, i=1,\hdots, n_u,
\ea
\eeq
 where $\xi_j(t) \in \bR$ is the radius of the maximum volume ball inscribed in $\mathcal{F}_{j}(t)$, $\alpha>0$ is a design variable, and $a_{ji}(t)$ and $b_{ji}(t)$ stand for the $i^{\text{th}}$ row of the matrix $A_j(t)$ and the vector $b_j(t)$.\\
Initially, at time step $t=0$, the matrix $H_{c}(0)$ can be taken as an arbitrary nonzero point inside the set $\mathcal{F}(0)$.
\subsection{Constrained predictive controller}\label{S:Control}
Let $u(k|t),\, k\in[t,t+N-1]$, $N\geq m$,  be the possible future control moves, where the notation $k|t$ indicates the prediction at step $k\geq t$ given the information at the current step $t$. Similarly, we define the vectors of future input increments $\Delta u(k|t), \, k\in [t,t+N-1]$ as:
\beq
\!\Delta u(k|t)\!=\!\begin{cases} u(t|t)\!-\!u(t-1)\!\! & \text{if}\, k=t\\
u(k|t)\!-\!u(k\!-\!1|t)\!\! & \text{if}\, t\!+\!1\leq k\leq t\!+\!N\!-\!1 \end{cases}
\eeq
Moreover, we define the future regressor vectors $\varphi(k|t)\in\bR^{n_um},\,k\in[t+1,t+N]$ as:
\begin{equation}\label{Eq:predicted_regressors}
\varphi(k|t)\!=\!\!\begin{cases} W\varphi(t)\!\!+\!\!Zu(t|t)\!\!\! & \text{if} \,\, k=t+1\\
                                                 W \varphi(k\!-\!1|t)\!\!+\!\!Zu(k\!-\!1|t) \!\!\!& \text{if} \,\, t\!+\!2\!\leq\! k\!\leq\! t\!+\!N.
\end{cases}
\end{equation}
In addition, we define the current prediction error $\hat{d}(t)\in \bR^{n_y}$ as the difference between the measured plant output and the one predicted by the nominal model at the time step $t$:
\beq\label{Eq:predicion_error}
\hat{d}(t)\doteq \tilde{y}(t)-H_c(t)\varphi(t).
\eeq
Then, we consider the following cost function:
\beq \label{Eq:cost_function}
\begin{aligned}
J(U,\tilde{y}(t),\varphi(t))\doteq \sum\limits_{k=t}^{t\!+\!N\!-\!1}\!&\!\left(\hat{y}(k\!+\!1|t)\!-\!y_{\text{des}}(k\!+\!1|t)\right)^TQ(\hat{y}(k\!+\!1|t)-\!y_{\text{des}}(k\!+\!1|t))\!+\! u(k|t)^TSu(k|t)\\&+\Delta u(k|t)^TR\Delta u(k|t),
\end{aligned}\!\!
\eeq
where:
\beq\label{Eq:predicted_output}
\hat{y}(k+1|t)=H_c(t)\varphi(k+1|t)+\hat{d}(t).
\eeq
In \eqref{Eq:cost_function}, $U=[u(t|t) \ldots u(t+N-1|t)]$ are the decision variables, while $\tilde{y}(t)$ and $\varphi(t)$ are  known parameters. $y_\text{des}(k|t),\,k\in[t+1,t+N]$, are the predicted values of the desired output. The introduction of the disturbance estimate $\hat{d}(t)$ in the cost function enables offset free tracking. Moreover, if the nominal model of the plant $H_c(t)$ were equal to the real plant, the measurement noise $v(t)$ were zero, and the output disturbance $d(t)$ were constant, for $N=T$, minimizing the cost function \eqref{Eq:cost_function} would be equivalent to minimizing the cost function of the control objective \eqref{Eq:problem_cost}.\\
Satisfaction of input constraints can be enforced by the following set of inequalities:
\beq\label{Eq:input_rate}
\ba{ll}
\begin{aligned}
C u(k|t) & \leq & d\\
L \Delta u(k|t) & \leq & f
\end{aligned} & \forall k\in[t,t+N-1].
\ea
\eeq
And the robust satisfaction of the output constraints can be achieved by enforcing the latter for all the plants inside the model set $\mathcal{F}(t)$ and all disturbance realizations:
\beq\label{Eq:output1}
E(H\varphi(k|t)\!+\!d)\leq p, \, \ba{l}\forall H\! \in\! \mathcal{F}(t),\, \forall d: -\epsilon_d\leq d\leq \epsilon_d,\\ \forall k\in[t\!+\!1,t\!+\!N] \ea
\eeq
The constraints \eqref{Eq:output1} are satisfied for anny disturbance realization if the following set of inequalities is satisfied: 
\beq\label{Eq:output}
EH\varphi(k|t)\!+\!\overline{d}\leq p, \, \forall H\! \in\! \mathcal{F}(t),\, \forall k\in[t\!+\!1,t\!+\!N],
\eeq
where $\overline{d}=[\overline{d}_1,\hdots, \overline{d}_{n_o}]^T$, and $d_l \in \bR,\, l=1,\hdots, n_o$ are given as:
\beq\label{Eq:iter_LP}
\overline{d}_l=\sum\limits_{j=1}^{n_y}|e_{lj}|\epsilon_{d_j},
\eeq
where $e_{lj}$ stand for the element of the $l^\text{th}$ row and $j^\text{th}$ column of the matrix $E$. However, using the constraints \eqref{Eq:output} would result in an infinite dimensional optimization problem, that is in general very hard to solve. The following result shows how \eqref{Eq:output} can be equivalently written in the form of linear constraints.
\begin{lemma}\label{L:constraint_equivalence_1}
The constraints (\ref{Eq:output}) are satisfied if and only if the following set of inequalities is feasible:
\small
\beq\label{Eq:relaxed_constraints1}
\left.\ba{l}\!\!
A(t)^T\lambda_l(k|t)=\left[\ba{c}
e_{l1}\varphi(k|t)\\
\vdots\\
e_{ln_y}\varphi(k|t)
\ea\right]\\
b(t)^T\lambda_l(k|t)\leq p_l\!-\!\overline{d}_l\\
\lambda_l(k|t)\geq \bold 0
\ea \right\}\ba{cc}\forall l=1,\hdots,n_o\\
\forall k\!\in\![t\!+\!1,t\!+\!N]\ea
\eeq
\normalsize
with
\small
\beq\label{Eq:explanation}
\begin{aligned}
A(t)&=\left[\ba{cccc} A_1(t)& \bold 0& \ldots&\bold 0\\
\bold 0& A_2(t)& \ldots& \bold 0\\
\vdots& \vdots & \ddots & \vdots \\
\bold 0&\bold 0& \ldots&A_{n_y}(t)\ea \right]\\
b(t)&=\left[\ba{c}
b_1(t)\\
\vdots\\
b_{n_y}(t)
\ea \right],
\end{aligned}
\eeq
\normalsize
where $\bold 0$ represents zero matrices of appropriate dimensions and $j^\text{th}$ column of the matrix $E$ and $p_l$ is the $l^\text{th}$ element of the vector $p$. In \eqref{Eq:relaxed_constraints1}, $\lambda_l(k|t) \in \bR^{r(t)}$ are additional decision variables, where $r(t)=\sum_{j=1}^{n_y} r_j(t)$.
\end{lemma}
\begin{proof}
We first note that, from the definition of the set $\mathcal{F}(t)$, it follows that constraints \eqref{Eq:output} are satisfied if and only if the following set of inequalities is satisfied:
\beq\label{Eq:LP1}
\ba{ll}
\gamma_l(k)\leq p_l\!-\!\overline{d}_l, & \ba{c} \forall l=1,\hdots,n_o\\\forall k\in [t+1,t+N]\ea\ea
\eeq
where
\beq\label{Eq:Linear_proggrams_1}
\gamma_l(k)=\max \limits_{A_j(t)H_j\leq b_j(t)}\sum\limits_{j=1}^{n_y}e_{lj}\varphi(k|t)^TH_j.
\eeq
For fixed values of the vectors $\varphi(k|t), \, k\in [t+1,t+N]$, by using the fact that the inequalities $A_j(t) H_j\leq b_j(t),\, j=1,\hdots, n_y$ form nonempty, closed and bounded convex sets (i.e. polytopes), we can write the dual of the LP \eqref{Eq:Linear_proggrams_1} as:
\bea
&\tilde{\gamma}_l(k)=\min \limits_{\lambda_l(k|t)}b(t)^T\lambda_l(k|t)\\
&\text{subject to} \nonumber\\\label{Eq:LP1c1}
&A(t)^T\lambda_l(k|t)=\left[ \ba{c}
e_{l1}\varphi(k|t)\\
\vdots\\
e_{ln_y}\varphi(k|t)\ea\right]\\\label{Eq:LP1c2}
&\lambda_l(k|t)\geq \bold 0 .
\eea
According to the strong duality theorem for LPs \cite{LP_book}, it holds that: $\gamma_l(k)=\tilde{\gamma}_l(k)$. Therefore, for any $\lambda_l(k|t)$ that satisfies the constraints (\ref{Eq:LP1c1}) and \eqref{Eq:LP1c2}, it holds that $\gamma_l(k) \leq b(t)^T\lambda_l(k|t)$. Hence the existence of $\varphi(k|t)$ and $\lambda_l(k|t)$ that satisfy the set of constraints (\ref{Eq:relaxed_constraints1}) guarantees that the constraints (\ref{Eq:LP1}) are also satisfied, which implies the satisfaction of the original constraints \eqref{Eq:output}. On the other hand if the constraints \eqref{Eq:output} are satisfied, then there exists $\gamma_l(k)$ satisfying \eqref{Eq:LP1}. Then by the strong duality theorem for LP, $\tilde{\gamma}_l(k)=\gamma_l(k)$ exists and hence the constraints \eqref{Eq:LP1c1} and \eqref{Eq:LP1c2} have to be feasible, which implies the feasibility of \eqref{Eq:relaxed_constraints1}. \hfill$\blacksquare$
\end{proof}
In order to be able to recursively satisfy the input and output constraints, we introduce an additional constraint on the terminal stage:
\beq\label{Eq:end_constraint}
\varphi(t+N|t)=W\varphi(t+N|t)+Zu(t+N-1|t).
\eeq
This means that we require the terminal regressor to correspond to a steady state.\\
For fixed values of $N$, $Q$, $S$ and $R$, we can now define the finite horizon optimal control problem (FHOCP) to be solved at time step $t$ (see step 3) of Algorithm \ref{Alg:RHCOP2}):
\beq \label{Eq:FHOCP}
\begin{aligned}
\min
\limits_{U} &J(U,\tilde{y}(t),\varphi (t))\\
&\text{subject to}\, \eqref{Eq:input_rate},\, \eqref{Eq:relaxed_constraints1}, \, \eqref{Eq:end_constraint},
\end{aligned}
\eeq
which is a quadratic program (QP), that can be efficiently solved.\\
The proposed adaptive control algorithm is such that if the problem \eqref{Eq:FHOCP} that is solved under Algorithm \ref{Alg:RHCOP2} has a feasible solution at time step $0$, then it is guaranteed to have a feasible solution $\forall t\geq 0$. This means that the proposed adaptive control algorithm guarantees the satisfaction of both input and output constraints $\forall t$.
\section{Adding an exploring property to the control algorithm}\label{S:for_identification}
The proposed adaptive control algorithm relies on the assumption that the discrepancy between the nominal and the actual model of the plant results in control inputs that are informative, such that over time the collected input-output data will reduce the size of the model set $\mathcal{F}(t)$ and therefore improve the accuracy of the identified plant model. The approach does not require a persistence of excitation assumption to avoid numerical problems, unlike other approaches based on least squares \cite{PE_ref2}. Nevertheless, in order to achieve good performance, the applied control inputs should be informative enough such that the model set $\mathcal{F}(t)$ becomes small as quickly as possible. A possible method to add an exploring property to the control algorithm is to split the calculation of the control input in two stages. In the first stage, the FHOCP \eqref{Eq:FHOCP} is solved as usual. The computed optimal control sequence and the knowledge of the model set $\mathcal{F}(t)$ are then used to calculate the upper bounds, along the chosen prediction horizon, on the absolute difference between all the possible future outputs of the plant and the nominal optimal output trajectory. In the second stage, by allowing these bounds to be inflated by a factor selected by the control designer, the sequence of control inputs can then be recalculated in order to improve the reduction in size of the model set. To this end, the control input is selected such that it approximately minimizes the volume of the polytopes that will form the model set at the next time step, while at the same time ensuring that the difference between the future trajectory of the plant outputs and the optimal one calculated in the first stage remains inside the selected bounds, and that the input and output constraints are satisfied.\\
With this approach, the relative importance of reference tracking and identification is automatically linked to the amount of information available on the system, which is represented by the size of the model set. In fact, if the model set $\mathcal{F}(t)$ is large, the input and trajectories computed at the second stage will be allowed to significantly deviate from the ones calculated in the first stage, in order to generate a control input that is informative and reduces the size of the model set $\mathcal{F}(t)$. On the other hand, if the uncertainty is small, the future plant output will be allowed to change only slightly from the first to the second stage.\\
To be more specific, let us consider the solution of the FHOCP \eqref{Eq:FHOCP}, which constitutes the first stage of the described approach. We denote the predicted regressor vectors and plant outputs obtained by solving \eqref{Eq:FHOCP} by $\varphi'(k|t)$ and $\hat{y}'(k|t),\, k\in[t+1,t+N]$. Then for the second stage, we compute the following quantities:
\beq\label{Eq:bounds1}
\begin{aligned}
\overline{\epsilon}_j(k|t)=&\max \left\{ \overline{y}_j(k|t)-\hat{y}'_j(k|t),\, \hat{y}'_j(k|t)-\underline{y}_j(k|t) \right \}, k\in [t+1,t+N],\,\, j=1,\hdots, n_y,
\end{aligned}
\eeq
where $\overline{\epsilon}(k|t)=[\overline{\epsilon}_1(k|t),\hdots, \overline{\epsilon}_{n_y}(k|t)]^T$, $\overline{\epsilon}_j(k|t)\geq 0$ denotes the maximal possible difference between the possible future output of the plant and the predicted output at time step $k$, and 
\beq
\ba{ccc}
\overline{y}_j(k|t)&=&\max\limits_{A_j(t)H_j\leq b_j(t)} H_j^T\varphi'(k|t)\\
 \underline{y}_j(k|t)&=& \min\limits_{A_j(t)H_j\leq b_j(t)} H_j^T\varphi'(k|t). \ea
\eeq
In addition, we define the matrix $\Phi(t+1|t)\in \bR^{n_um\times n_um}$ that depends on the $n_um$ past regressor vectors and the first future regressor vector as:
\beq
\Phi(t+1|t)=\left[\ba{cccc} \varphi(t-n_um+1) &\hdots & \varphi(t) & \varphi(t+1|t) \ea \right].
\eeq
 This matrix can be indirectly related to the size of the polytopes that will form the model set at the next time step $\mathcal{F}(t+1)$ by the following result which is taken from \cite{id_aspect}.
\begin{lemma}\label{L:volume}
Each of the polytopes $\mathcal{F}_j(t+1),\, j=1,\hdots,n_y$ obtained by using the polytopic update of the form \eqref{Eq:Normal_update} for a sequence of control input vectors that define the regressor vectors forming the matrix $\Phi(t+1|t)$ is guaranteed to have the volume smaller than $\frac{\left(2(\epsilon_{d_j}+\epsilon_{v_j})\right)^{n_um}}{\left| \det\left( \Phi(t+1|t)\right)\right|}$.\hfill$\blacksquare$
\end{lemma}
The input to be applied to the plant is selected within a set of all the control inputs that satisfy input and output constraints and that keep all of the possible predicted output trajectories inside an interval obtained by scaling the values of $\overline{\epsilon}(k|t)$, centered at the trajectory $\hat{y}'(k|t),\, k\in[t+1,t+N]$.  In order to improve the knowledge on the system, we need a suitable cost function that penalizes the size of the model set. By considering the result of \cite{id_aspect}, we chose to consider $|\det(\Phi(t+1|t))|$ and compute an input that increases its value. Therefore, the optimization problem to be solved at the second stage of the control input calculation is given as:
\beq\label{Eq:FHOCP2_cost}
\begin{aligned}
&\max \limits_U \left|\det \Phi(t+1|t)   \right|\\
&\text {subject to}\\
&\eqref{Eq:input_rate}, \, \eqref{Eq:end_constraint},\\
&\left.\ba{lll}
H\varphi(k|t)&\!\leq\!& \hat{y}'(k|t)\!+\!r\overline{\epsilon}(k|t)\\
H\varphi(k|t)&\!\geq\! &\hat{y}'(k|t)\!-\!r\overline{\epsilon}(k|t)\\
EH\varphi(k|t)\!+\!\overline{d}&\! \leq\!&\! p
\ea \right\}\ba{ll}\forall H\in \mathcal{F}(t)\\ \forall k\in[t\!+\!1, t\!+\!N] \ea
\end{aligned}
\eeq
where $r\in \bR$, $r\geq 1$ is a design parameter, selected by the control designer, that determines the allowed augmentation of the intervals containing the predicted outputs between the first and second stages. This parameter indicates by how much the bounds \eqref{Eq:bounds1} are allowed to be inflated. Problem \eqref{Eq:FHOCP2_cost} is a non convex, infinite dimensional program that is in general very difficult to solve. However, in this specific case the infinite dimensional constraints can be reformulated into a set of linear constraints. In fact, if we define the matrix $E''$ and the vectors $\overline{d''}$ and $p''(k|t),\, k\in [t+1,t+N]$ as:
\beq\label{Eq:redefine_matrices}
E''=\left[ \ba{c} I\\ -I\\ E \ea \right],\, \overline{d''}=\left[\ba{c} \bold{0}\\ \bold{0}\\ \overline{d} \ea   \right],\, p''(k|t)=\left[ \ba{c} \hat{y}'(k|t)\!+\!r\overline{\epsilon}(k|t)\\-\hat{y}'(k|t)\!+\!r\overline{\epsilon}(k|t)\\p\ea \right],
\eeq
where $I$ and $\bold{0}$ respectively denote the identity and zero matrices of apropriate size, then the infinite dimensional constraints in \eqref{Eq:FHOCP2_cost} can be rewritten as
\beq
E''H\varphi(k|t)\!+\!\overline{d''}\leq p''(k|t),\, \forall H\in \mathcal{F}(t),\, k\in[t\!+\!1,t\!+\!N],
\eeq
and these constraints can be written in the form of linear constraints by using the Lemma \ref{L:constraint_equivalence_1}. In addition, we note that from the definition of the matrix $\Phi(t+1|t)$ and the regressor vector $\varphi(t+1|t)$, it follows that the determinant of the matrix $\Phi(t+1|t)$ can be expressed as a linear function of the first future control input as: 
\beq\label{Eq:linear_m}
\det\left(\Phi(t+1|t)\right)=k^T u(t|t)+n,
\eeq
where the derivation of $k \in \bR^{n_u}$ and $n \in \bR$ is given in the Appendix \ref{S:alpha_beta}. Therefore, the optimization problem \eqref{Eq:FHOCP2_cost} can be solved by first solving the following two LPs:
\bea\label{Eq:cost1}
J_1=\max \limits_{u(k|t)} k^T u(t|t)\\\label{Eq:cost2}
J_2=\min \limits_{u(k|t)} k^T u(t|t) \label{Eq:constraints}
\eea
\beq
\ba{l}
\text{subject to}\\
\eqref{Eq:input_rate}, \, \eqref{Eq:end_constraint},\\
\left.\ba{l}\!\!
A(t)^T\lambda_l(k|t)=\left[\ba{c}
e_{l1}''\varphi(k|t)\\
\vdots\\
e_{ln_y}''\varphi(k|t)
\ea\right]\\
b(t)^T\lambda_l(k|t)\leq p_l''(k|t)\!-\!\sum\limits_{j=1}^{n_y}e_{lj}''\epsilon_{d_j}\\
\lambda_l(k|t)\geq \bold 0
\ea\!\!\! \right\}\ba{cc}\forall l=1,\!\hdots,\!n_o\!+\!2n_y\\
\forall k\!\in\![t\!+\!1,t\!+\!N]\ea
\ea
\eeq
where $e_{lj}''$ stand for the element of the $l^\text{th}$ row and $j^\text{th}$ column of the matrix $E''$ and $p_l''(k|t)$ is the $l^\text{th}$ element of the vector $p''(k|t)$. The matrix $A(t)$ and the vector $b(t)$ are defined as in \eqref{Eq:explanation}. The solution to \eqref{Eq:FHOCP2_cost} can then be obtained by comparing the cost function values of the two LPs: if $|J_1+n|\leq |J_2+n|$, then the solution of the LP with the cost function \eqref{Eq:cost1} is the solution of the original problem and vice-versa. Initially, at time step $t=0$, the regressor vectors that form the matrix $\Phi(1|0)$ can be constructed from the constraints that form the initial model set $\mathcal{F}(0)$.\\
The proposed extension is computationally tractable as it requires the additional solution of two LPs. Moreover, the optimization problem that is solved in the second stage of the control input calculation is guaranteed to be recursively feasible, since the solution obtained in the first stage is always a feasible solution for the second stage. Due to this, the proposed extension allows us to retain the guarantee of constraint satisfaction for all time.  
\begin{remark}\label{R:note_on_complexity}
The cost function \eqref{Eq:FHOCP2_cost} only approximately minimizes the volume of the polytopes that will form the model set in the next time step. In order to have exact minimization, all the collected control inputs and the effects of the limited polytopic update algorithm would need to be taken into account. However, the resulting optimization problem would then be computationally intractable. In addition, the cost function minimizes only the volume of the polytopes in the next time step and not over the whole prediction horizon, as this would also lead to a computationally intractable optimization problem. Despite these approximations, the proposed modification gives good results in practice, as shown by the numerical example.
\end{remark}
\section*{Appendix}
\subsection{Definition of the matrices in \eqref{Eq:recursive_regressor}} \label{S:Matrices_reg}        
For the case when impulse response parametrization is used ($\mathcal{L}_k(a,q)$ is given by \eqref{Eq:impulse_response}), we define the matrix $w\in \bR^{n_u \times n_u}$ with the following structure:
\begin{equation*}
w=\left[ \ba{ccccc} 0 & 0 & \hdots & 0 & 0\\
                                a & 0 & \hdots & 0 & 0\\
                                0 & a & \hdots & 0 & 0\\
                                \vdots & \vdots & \ddots & \vdots & \vdots\\
                               0 & 0 & \hdots & a & 0
\ea \right] \in \bR^{m\times m}.
\end{equation*}
Based on this, the matrix $W$ is given by:
\begin{equation}\label{Eq:W}
W=\left[\ba{cccc} w &\bold{0} & \hdots & \bold{0}\\
              \bold{0} & w & \hdots & \bold{0}\\
              \vdots & \vdots & \ddots & \vdots\\
               \bold{0} & \bold{0} & \hdots & w
\ea\right] \in \bR^{n_um \times n_um},
\end{equation}
where $\bold{0}$ denotes the matrix of all zeros with appropriate dimension.
Similarly, in order to define the matrix $Z\in \bR^{n_um\times n_u}$, we first define the vector $z=[a,0,\hdots,0]^T\in \bR^m$. Then we can write the matrix $Z$ as:
\begin{equation}\label{Eq:Z}
Z=\left[\ba{cccc} z & \bold{0} & \hdots & \bold{0}\\
                             \bold{0} & z & \hdots & \bold {0}\\
                              \vdots & \vdots & \ddots & \vdots \\
                              \bold{0} & \bold{0} & \hdots & z \ea   \right]\in \bR^{n_um\times n_u},
\end{equation}
where $\bold{0}$ denotes the vector of all zeros of dimension $m$.\\
For the case when Laguerre basis functions are used for parametrization ($\mathcal{L}_k(a,q)$ is given by \eqref{Eq:Laguere}), we define the matrix $w \in \bR^{m\times m}$ as:
\begin{equation*}
w=\left[\ba{ccccc} a & 0 & 0 & \hdots & 0\\
                                                    1-a^2 & a & 0 &\hdots & 0\\
                                                     (-a)(1-a^2) & 1-a^2 & a & \hdots & 0\\
                                                     \vdots & \vdots & \vdots & \ddots & \vdots\\
                                                     (-a)^{m-2}(1-a^2) & (-a)^{m-3}(1-a^2)& \hdots &\hdots& a \ea   \right].
\end{equation*}
The matrix $W$ is then defined as in \eqref{Eq:W}. The vector $z\in \bR^m$ is defined as:
\begin{equation*}
z=\sqrt{1-a^2}\left[\ba{c} 1\\-a\\(-a)^2\\ \hdots \\ (-a)^{m-1} \ea  \right]\in \bR^m,
\end{equation*}
and $Z$ is given as in \eqref{Eq:Z}.\\
For the Kautz basis functions, the matrices can be derived in a similar way, but have a more complex structure. In the case when the combination of impulse response and basis function parametrization is used, the matrices can easily be derived along the lines presented here.
\subsection{Definition of the matrices in \eqref{Eq:linear_m}}\label{S:alpha_beta}
In order to define $k \in \bR^{n_u}$ and $n \in \bR$, we first shortly recall the definition of the matrix determinant. For a square matrix $A\in \bR^{N\times N}$, the determinant is defined as:
\beq\label{Eq:det_def}
\det (A)=\sum \limits_{\sigma \in S_N}\sgn(\sigma) \prod \limits_{i=1}^N A_{i,\sigma_i}, 
\eeq 
where $S_N$ denotes the set of all the permutations of the integers $\{1,\hdots,N\}$, $\sgn (\sigma)$ returns $1$ if the permutation $\sigma$ is even and $-1$ if the permutation is odd and $A_{i,\sigma_i}$ denotes the element of the matrix $A$ in the $i^{\text{th}}$ row and the column defined by the $i^{\text{th}}$ element in the permutation $\sigma$.\\
We rewrite the matrix $\Phi(t+1|t)$ as $\Phi(t+1|t)=\left[\Phi',\, W\varphi(t)+Zu(t|t)\right]$, where
\begin{equation*}
\Phi'=\left[ \ba{ccc}  \Phi_{1,1} &\hdots & \Phi_{1,n_um-1}\\ \vdots & \ddots & \vdots\\ \Phi_{n_um,1}& \hdots & \Phi_{n_um,n_um-1} \ea \right],
\end{equation*}
with the elements $\Phi_{p,q}\in \bR,\,\, p=1,\hdots,n_um,\, q=1,\hdots,n_um-1$, denoting the corresponding elements of the applied regerssor vectors $\varphi(k),\, k=t-n_um+1,\hdots, t-1$. In addition, we denote by $S_{n_um}^i,\, i=1,\hdots, n_um$ the set of all the permutations of the integers $\{1,\hdots,n_um\}$ that have $n_um$ at the $i^{\text{th}}$ position. Based on this we define the vector $V \in \bR^{n_um}$ as:
\beq\label{Eq:medjurezultat2}
V=\left[ \ba{c}  \sum \limits_{\sigma \in S_{n_um}^1} sgn(\sigma) \prod\limits_{i=2}^{n_um}\Phi_{i,\sigma_i}\\ \sum \limits_{\sigma \in S_{n_um}^2} sgn(\sigma) \prod\limits_{i=1,\,i\neq 2}^{n_um}\Phi_{i,\sigma_i}\\ \vdots \\\sum \limits_{\sigma \in S_{n_um}^{n_um}} sgn(\sigma) \prod\limits_{ i=1,\, i\neq n_um}^{n_um}\Phi_{i,\sigma_i}\ea \right].
\eeq
Then we can express $k$ and $n$ by using $V$ as:
\beq\label{Eq:alpha_beta_def}
\begin{aligned}
k^T&=V^TZ\\
n&=V^TW\varphi(t).
\end{aligned}
\eeq

\end{document}